\documentclass[letterpaper, 10pt, conference]{Classes/ieeeconf}
\IEEEoverridecommandlockouts
\usepackage{Packages/decar-common}
\usepackage{Packages/decar-post}
\usepackage{Packages/myFormat}

\title{\LARGE \bf A Passivity-Based Analysis of First-Order Momentum-Based Methods}

\author{Sepehr Moalemi$^{1}$ and James Richard Forbes$^{2}$
    \thanks{$^{1}$Sepehr Moalemi {\tt\small moalemi@umich.edu} is with the Department of Electrical and Computer Engineering, University of Michigan, Ann Arbor, MI 48109, USA\@. $^{2}$James Richard Forbes {\tt\small james.richard.forbes@mcgill.ca} is with the Department of Mechanical Engineering, McGill University, 817 Sherbrooke St. W., Montreal, QC H3A 0C3, Canada.}
}

\begin{document}
    \maketitle
    \thispagestyle{empty}
    \pagestyle{empty}

    \fontdimen16\textfont2=\fontdimen17\textfont2
    \fontdimen13\textfont2=5pt
    \begin{abstract}
    This paper presents a discrete-time passivity-based analysis of first-order momentum-based methods for a class of functions whose gradient has lower and upper sector bounds of \(0\) and \(L\), respectively. Through a loop transformation, it is shown that momentum-based methods can be represented as a passive controller in negative feedback with an output strictly passive (OSP) system. The weak passivity theorem is then used to derive explicit hyperparameter conditions under which the shifted gradient asymptotically vanishes. Under an additional assumption that requires the existence of a unique stationary point and excludes arbitrarily small gradients far from that point, convergence of the iterates to the global minimizer is established.
\end{abstract}

    \section{Introduction}
Consider the unconstrained optimization problem
\begin{equation} \label{eq:optimization_problem}
    \min_{\mbf{x} \in \mathbb{R}^{n}} f(\mbf{x}),
\end{equation}
where \(\mbox{$f: \mathbb{R}^n \rightarrow \mathbb{R}$}\) is a continuously differentiable function. Assume there exists a minimizer \(\mbf{x}^\ast\) such that \(\nabla f(\mbf{x}^\ast) = \mbf{0}\). The function \(f\) is said to have a sector-bounded gradient if there exist \(m \in \mathbb{R}_{\geq 0}\) and \(L \in \mathbb{R}_{>0}\), with \(m \leq L\), such that
\begin{equation} \label{eqn:sector_bound}
    \begin{split}
        \hspace{-5pt}
        \left\langle\mbf{x} - \mbf{x}^\ast, \nabla f(\mbf{x})\right\rangle
        \geq
        \frac{mL}{m + L} \mleft\| \mbf{x} - \mbf{x}^\ast \mright\|_2^2
        +
        \frac{1}{m + L} \mleft\| \nabla f(\mbf{x}) \mright\|_2^2,
    \end{split}
\end{equation}
for all \(\mbf{x} \in \mathbb{R}^n\)~\cite[Proposition 5]{lessard_recht_iqc}. The set of continuously differentiable functions with a minimizer satisfying \cref{eqn:sector_bound} is denoted by \(\mathcal{F}_{m, L}\). This class contains non-convex functions and is more general than the set of \(m\)-strongly convex and \(L\)-smooth functions, denoted by \(\mathcal{S}_{m, L}\)~\cite{hu_lessard}.

Many first-order momentum-based methods for solving \cref{eq:optimization_problem} can be written as
\begin{equation}\label{eqn:recursion_momentum_augmented}
    \begin{aligned}
        \mbf{x}^{k + 1} &= \mbf{x}^{k} + \beta \mleft( \mbf{x}^{k} - \mbf{x}^{k-1} \mright) - \alpha \mbf{u}^{k},\\%
        \mbf{y}^{k} &= \mbf{x}^{k} + \gamma \mleft( \mbf{x}^{k} - \mbf{x}^{k-1} \mright),\\%
        \mbf{u}^{k} &= \nabla f(\mbf{y}^{k}),
    \end{aligned}
\end{equation}
for \(k \in \mathbb{Z}_{\geq 0}\) with initial conditions \(\mbf{x}^{-1}\) and \(\mbf{x}^{0}\), where \(\alpha \in \mathbb{R}_{> 0}\) is the stepsize, \(\beta \in \mathbb{R}_{\geq 0}\) is the momentum parameter, and \(\gamma \in \mathbb{R}_{\geq 0}\) is the look-ahead parameter. The gradient descent~(GD) method is recovered for \(\beta = \gamma = 0\), Polyak's heavy-ball~(HB) method~\cite{polyak_hb} is recovered for \(\beta > 0\) and \(\gamma = 0\), and Nesterov's accelerated gradient (NAG) method~\cite{Nesterov} is recovered for \(\beta > 0\) and \(\gamma = \beta\). Such recursions have been studied extensively through discrete-time feedback and robust control viewpoints.

In the control-theoretic analysis of first-order methods, the sector bounds \cref{eqn:sector_bound} on the gradient are the main assumptions used to guarantee the stability of the feedback interconnection of the method and the gradient. Importantly, for the case of \(f \in \mathcal{S}_{m,L}\), it is possible to have \(\bar{m} \geq m\) and \(\bar{L} \leq L\) such that \(f \in \mathcal{F}_{\bar{m}, \bar{L}}\), meaning that the sector bounds on the gradient can be tightened. In \cite{ugrinovskii}, the circle criterion is used to guarantee the global convergence of the HB method for functions in \(\mathcal{F}_{m, L}\). In~\cite{alex_petersen}, this framework is extended by considering \(\gamma \neq 0\) in \cref{eqn:recursion_momentum_augmented}, yielding the generalized accelerated gradient (GAG) method. In~\cite{hu_lessard_dissipativity_nag}, for \(f \in \mathcal{S}_{m, L}\), dissipativity theory is used to analyze the NAG method by linking convergence-rate certificates to energy dissipation, enabling Lyapunov-function construction via semidefinite programming. Integral quadratic constraints (IQCs) are used in \cite{triple_momentum} to introduce the triple momentum (TM) method for functions in \(\mathcal{S}_{m, L}\). This result is extended in \cite{van2025fastest} for functions in \(\mathcal{S}_{m, L}\) that are also twice continuously differentiable, resulting in the \(C^2\)-momentum (C2M) method. Finally, for \(f \in \mathcal{F}_{m,L}\), \cite{moalemi_forbes_gd_passive_acc} shows that, after a loop transformation, the GD method can be represented as a passive controller in negative feedback with a very strictly passive system, after which the passivity theorem guarantees global convergence of the method.

Importantly, the analyses in \cite{ugrinovskii, alex_petersen, hu_lessard_dissipativity_nag, triple_momentum, van2025fastest, moalemi_forbes_gd_passive_acc} assume \(m>0\), whether through \(f \in \mathcal{S}_{m,L}\) or \(f \in \mathcal{F}_{m,L}\). In practice, estimating \(m\) can be difficult, or \(m\) may be zero. This paper extends the passivity-based analysis of GD in~\cite{moalemi_forbes_gd_passive_acc} to first-order momentum-based methods applied to functions in \(\mathcal{F}_{0,L}\). Explicit hyperparameter conditions are derived that guarantee asymptotic vanishing of the shifted gradient via the weak passivity theorem. In general, this does not imply convergence of the iterates to the global minimizer. As such, a stronger convergence result is established under an additional assumption requiring the existence of a unique stationary point and excluding arbitrarily small gradients far from that point.

The remainder of this paper is organized as follows. \Cref{sec:preliminaries} introduces the preliminaries and the control interpretation of first-order momentum-based methods. \Cref{sec:main_results} presents the loop transformation, the passivity-based analysis, and the main stability results. \Cref{sec:discussion} discusses the resulting convergence conditions and compares the GD, HB, and NAG cases with existing results in the literature. Numerical results are presented in \Cref{sec:simulation}, followed by closing remarks in \Cref{sec:closing_remarks}.

    \section{Preliminaries} \label{sec:preliminaries}
\subsection{Notation}
    A positive definite matrix is denoted by \(\mbf{A} \succ 0\) and a negative semidefinite matrix is denoted by \(\mbf{A} \preceq 0\). For a function \(\mbf{u} : \mathbb{Z}_{\geq 0} \to \mathbb{R}^{n}\) and \(T \in \mathbb{Z}_{\geq 0}\), its truncation, \(\mbf{u}_T\), is defined as \(\mbf{u}_T(k) = \mbf{u}(k)\) for \(k < T\) and \(\mbf{u}_T(k) = \mbf{0}\) for \(k \geq T\). The truncated inner product over the discrete-time interval \(\mathcal{T} = \mleft\{0, 1, \hdots, T-1 \mright\}\) is defined as \(\langle \mbf{u}, \mbf{y}\rangle_T = \langle\mbf{u}_T, \mbf{y}_T\rangle = \sum_{k \in \mathcal{T}} \mbf{u}^{\trans}(k) \mbf{y}(k), \forall T \in \mathbb{Z}_{> 0}\). Define \(\ell_{2e}\) as the set of all functions \(\mbf{u} : \mathbb{Z}_{\geq 0} \to \mathbb{R}^{n}\) such that \(\mleft\| \mbf{u} \mright\|_{2T}^2 = \langle \mbf{u}, \mbf{u}\rangle_T < \infty\), for all \(T \in \mathbb{Z}_{>0}\). Moreover, \(\mbf{u} \in \ell_{2}\) if \(\mleft\| \mbf{u} \mright\|_{2}^2 = \langle \mbf{u}, \mbf{u}\rangle < \infty\). The operator \(\bm{\mathcal{G}} : \ell_{2e} \to \ell_{2e}\) is \(\ell_{2}\)-stable if \(\bm{\mathcal{G}}\mbf{u} \in \ell_{2}\) for all \(\mbf{u} \in \ell_{2}\). The identity and zero matrices are denoted by \(\eye\) and \(\mbf{0}\), respectively.

\subsection{Control Interpretation of Optimization Algorithms}
    This section presents the discrete-time feedback representation of first-order optimization algorithms as discussed in~\mbox{\cite[Section 2]{lessard_recht_iqc}}.

    \begin{figure}[t]
        \centering
        \vspace{1pt}
        \includegraphics{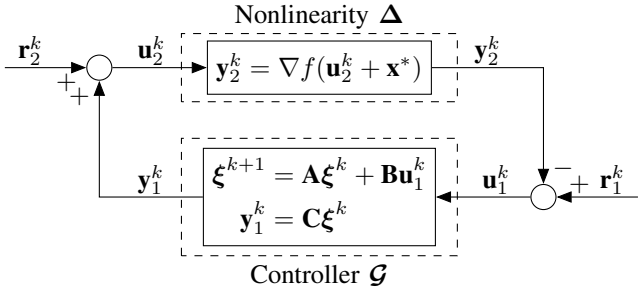}
        \caption{The negative feedback interconnection between the controller \(\bm{\mathcal{G}}\), whose state-space realization is given in \cref{eqn:momentum_controller_minimal_realization}, and the shifted-gradient nonlinearity, which maps zero inputs to zero outputs.}
        \label{fig:momentum_controller}
        \vspace{-10pt}
    \end{figure}
    Using the shifted-augmented state
    \begin{equation}\label{eqn:augmented_state}
        \mbs{\xi}^{k} = 
        \begin{bmatrix}
            \mbf{x}^{k} - \mbf{x}^\ast \\
            \mbf{x}^{k-1} - \mbf{x}^\ast
        \end{bmatrix},
    \end{equation}
    where \(\mbf{x}^\ast\) is a minimizer of \(f\), the recursion in \cref{eqn:recursion_momentum_augmented} can be interpreted as a static memoryless nonlinearity \(\mbs{\Delta}\) in negative feedback with a linear time-invariant (LTI) dynamical system \(\bm{\mathcal{G}}\), as shown in \Cref{fig:momentum_controller}. A state-space realization of the controller \(\bm{\mathcal{G}}\) is given by
    \begin{equation} \label{eqn:momentum_controller_minimal_realization}
        \left[
            \begin{array}{c|c}
                    \mbf{A} & \mbf{B}   \\ \noalign{\hrule}
                    \mbf{C} & \mbf{D}
            \end{array}
        \right]
        =
        \left[
            \begin{array}{cc|c}
                    1 + \beta  &  -\beta  & \alpha  \\
                    1            & 0      & 0    \\ \noalign{\hrule}
                    1 + \gamma &  -\gamma & 0
            \end{array}
        \right] \otimes \eye,
    \end{equation}
    where \(\otimes\) denotes the Kronecker product. Moreover, \(\mbs{\Delta}\) maps \(\mbf{u}_2\) to \(\mbf{y}_2\) according to \(\mbf{y}_{2}^{k} = \nabla f(\mbf{u}_{2}^{k} + \mbf{x}^\ast)\), so that \(\mbf{u}_{2}^{k} = \mbf{0}\) implies \(\mbf{y}_{2}^{k} = \nabla f(\mbf{x}^\ast) = \mbf{0}\).
    
    Notably, for \(\gamma = 0\) and \(\gamma = \beta\), the controller \(\bm{\mathcal{G}}\) in \cref{eqn:momentum_controller_minimal_realization} reduces to the HB and NAG controllers, respectively. Moreover, for particular choices of \(\alpha\), \(\beta\), and \(\gamma\), the controller \(\bm{\mathcal{G}}\) can also represent methods such as TM, GAG, and C2M.

\subsection{Weak Passivity Theorem and Definition}
    This section presents the weak passivity theorem within the context of the feedback interconnection of the controller \(\bm{\mathcal{G}}\) and the shifted gradient \(\mbs{\Delta}\) as shown in \Cref{fig:momentum_controller}.
    \begin{definition}[Passivity~{\cite{feedback_systems}}] \label{def:passivity}
        Consider a square system represented by an operator \(\bm{\mathcal{H}}:\ell_{2e}\to\ell_{2e}\) that maps the input \(\mbf{u}\) to the output \(\mbf{y}\). The system \(\bm{\mathcal{H}}\) is 
        \begin{enumerate}[label=\textit{\roman*})]
            \item {%
            \emph{passive} if \(\exists \mu  \in \mathbb{R}_{\leq 0} \) such that
            \begin{equation*}
                \langle \mbf{u}, \mbf{y}\rangle_T 
                \geq
                \mu
                , \quad \forall \mbf{u} \in \ell_{2e},\,\forall T \in \mathbb{Z}_{>0},
            \end{equation*}
            }%
            \item {%
            \emph{output strictly passive (OSP)} if \(\exists \mu \in \mathbb{R}_{\leq 0} \) and \(\exists \varepsilon \in \mathbb{R}_{>0} \) such that
            \begin{equation*}
                \langle \mbf{u}, \mbf{y}\rangle_T 
                \geq
                \mu + \varepsilon \mleft\| \mbf{y} \mright\|^{2}_{2T}
                , \quad \forall \mbf{u} \in \ell_{2e},\,\forall T \in \mathbb{Z}_{>0}.
            \end{equation*}
            }%
        \end{enumerate}
    \end{definition}
    \vspace{2pt}
    A special case of the passivity theorem, known as the weak passivity theorem, is given below.
    \begin{theorem}[Weak passivity theorem~{\cite{van_der_schaft}}]\label{thrm:weak_passivity}
        Consider two systems \(\bm{\mathcal{G}} : \ell_{2e} \to \ell_{2e}\) and \(\mbs{\Delta} : \ell_{2e} \to \ell_{2e}\) in negative feedback as per \Cref{fig:momentum_controller} with \(\mbf{r}_1 = \mbf{0}\). If \(\bm{\mathcal{G}}\) is passive, \(\mbs{\Delta}\) is OSP, and \(\mbf{r}_2 \in \ell_{2}\), then \(\mbf{y}_2 \in \ell_{2}\).
    \end{theorem}

    \section{Analysis} \label{sec:main_results}
This section extends the passivity-based analysis of the GD method in~\cite{moalemi_forbes_gd_passive_acc} to first-order momentum-based methods of the form~\cref{eqn:recursion_momentum_augmented}.
\subsection{Loop Transformation}
    \begin{figure}[t]
        \centering
        \vspace{1pt}
        \includegraphics{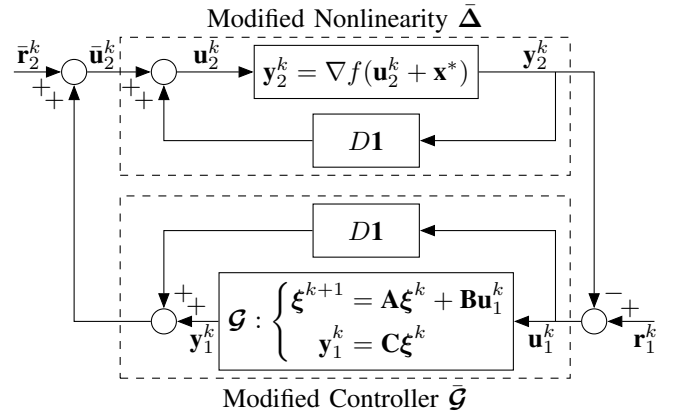}
        \caption{Loop transformation of the negative feedback interconnection in \Cref{fig:momentum_controller}. The loop transformation introduces a feedthrough term \(D\eye\), resulting in the modified controller \(\bar{\bm{\mathcal{G}}}\), the modified nonlinearity \(\bar{\mbs{\Delta}}\), and \(\bar{\mbf{r}}_2^k = \mbf{r}_2^k - D\mbf{r}_1^k\).}
        \vspace{-20pt}
        \label{fig:momentum_loop_transformation}
    \end{figure}
    Since the controller \(\bm{\mathcal{G}}\) in \cref{eqn:momentum_controller_minimal_realization} is strictly proper (\(\mbf{D} = \mbf{0}\)), it can never be passive~\mbox{\cite[Remark 2.7]{Byrnes_Lin_Losslessness}}. Therefore, the loop transformation in \Cref{fig:momentum_loop_transformation} is performed to introduce a constant feedthrough term \(\mbf{D} = D\eye\), where \(D \in \mathbb{R}_{>0}\). This loop transformation results in the modified controller \(\bar{\bm{\mathcal{G}}}\) and the modified nonlinearity \(\bar{\mbs{\Delta}}\), but does not change the closed-loop dynamics of the system, meaning \Cref{fig:momentum_loop_transformation} still represents the recursion in \cref{eqn:recursion_momentum_augmented}.
\subsection{Passivity of \texorpdfstring{\(\bar{\bm{\mathcal{G}}}\)}{G} and \texorpdfstring{\(\bar{\mbs{\Delta}}\)}{Delta}}
    This section outlines the passivity properties of \(\bar{\bm{\mathcal{G}}}\) and \(\bar{\mbs{\Delta}}\) in two separate lemmas. First, it is shown that \(\bar{\mbs{\Delta}}\) is OSP for any function in \(\mathcal{F}_{0, L}\), provided \(D < 1/L\). Then, it is shown that \(\bar{\bm{\mathcal{G}}}\) is passive, provided \(D\) is sufficiently large.

    The following lemma is adapted from \cite[Lemma 4]{moalemi_forbes_gd_passive_acc} for the case of functions in \(\mathcal{F}_{0, L}\).
    \begin{lemma} \label{lem:modified_nonlinearity_passivity}
        Consider the modified nonlinearity \(\bar{\mbs{\Delta}}\) in \Cref{fig:momentum_loop_transformation}, where \(f \in \mathcal{F}_{0, L}\). 
        Given a feedback term \(D\eye\) with \(D \in \mathbb{R}_{>0}\), if \(D < 1/L\), the positive feedback interconnection of \(\mbs{\Delta}\) and \(D\eye\) results in an OSP system \(\bar{\mbs{\Delta}}\). 
    \end{lemma}
    \begin{proof}
        From \cref{eqn:sector_bound} with \(m = 0\), it follows that
        \begin{align*}
            \left\langle \bar{\mbf{u}}_2, \mbf{y}_2 \right\rangle_{T}
            =
            \left\langle \mbf{u}_2, \mbf{y}_2 \right\rangle_{T} - D\mleft\| \mbf{y}_2 \mright\|_{2T}^2
            \geq \mleft( \frac{1}{L}  - D \mright) \mleft\| \mbf{y}_2 \mright\|_{2T}^2.
        \end{align*}
        Therefore, if \(D < 1/L\), \(\bar{\mbs{\Delta}}\) is OSP with \(\varepsilon = 1/L - D\) and \(\mu = 0\).\hspace*{\fill}~\QED\par\endtrivlist\unskip
    \end{proof}

    The following lemma provides a sufficient condition for the modified controller \(\bar{\bm{\mathcal{G}}}\) to be passive. The proof is provided in \hyperref[appendix:proof_of_passivity_of_G_bar]{Appendix A}.
    \begin{lemma}\label{lemma:passivity_of_G_bar}
        Given a feedthrough term \(D\eye\) with \(D \in \mathbb{R}_{>0}\), \(\alpha \in  \mathbb{R}_{>0}\), and \(\beta, \gamma \in \mathbb{R}_{\geq 0}\), the modified controller \(\bar{\bm{\mathcal{G}}}\) in \Cref{fig:momentum_loop_transformation} is passive if
        \begin{equation}
                D \geq
            \begin{cases}
                \dfrac{\alpha \mleft(2 \gamma + 1\mright)}{2 \mleft(\beta + 1\mright)}, & \beta < \dfrac{\gamma}{\gamma + 1},\\[2ex]%
                \dfrac{\alpha }{\mleft(1 - \beta\mright)^2} - \dfrac{\alpha \mleft(2 \gamma + 1\mright)}{2 \mleft(1 - \beta\mright)}, & \dfrac{\gamma}{\gamma + 1} < \beta < 1.
            \end{cases}
        \end{equation}
    \end{lemma}
    \vspace{2pt}
    To further examine the case of \(\beta = \frac{\gamma}{\gamma + 1}\), consider the transfer function of \(\bar{\bm{\mathcal{G}}}\) given by \(\bar{\mbf{G}}(z) = \bar{g}(z) \otimes \eye\), where
    \begin{equation}
        \bar{g}(z) = D + \alpha \mleft( \gamma + 1  \mright) \frac{\mleft( z - \frac{\gamma }{\gamma + 1}  \mright) }{\mleft( z - 1 \mright) \mleft( z -  \beta  \mright) }.
    \end{equation}
    \begin{remark}\label{rmk:case_of_equality_for_beta_gamma}
        For \(\beta = \frac{\gamma}{\gamma + 1}\), \(\bar{g}(z)\) has a pole-zero cancellation, resulting in \(\bar{g}(z) = D + \alpha (\gamma + 1)/(z - 1)\). This recovers the transfer function of the modified GD controller with stepsize \(\alpha (\gamma + 1)\) and feedthrough \(D\), which is positive real by \cite[Lemma 2]{moalemi_forbes_gd_passive_acc} if \(D \geq \alpha (\gamma + 1)/2\). Therefore, for \(\beta = \frac{\gamma}{\gamma + 1}\), \(\bar{\bm{\mathcal{G}}}\) is passive if \(D \geq \alpha (\gamma + 1)/2\). Notably, this condition is consistent with the results of \Cref{lemma:passivity_of_G_bar} for \(\beta < \frac{\gamma}{\gamma + 1}\) and \(\beta > \frac{\gamma}{\gamma + 1}\) in the limit as \(\beta \to \frac{\gamma}{\gamma + 1}\).
    \end{remark}

    Importantly, the results of \Cref{lemma:passivity_of_G_bar} are not dependent on the properties of \(f\), meaning for any choice of \(\alpha\), \(\gamma\), and \(\beta < 1\), there exists a sufficiently large \(D\) such that \(\bar{\bm{\mathcal{G}}}\) is passive. Moreover, using an approach similar to that used in the proof of \Cref{lemma:passivity_of_G_bar}, it is possible to determine the minimum feedthrough \(D\) required to make \(\bar{\bm{\mathcal{G}}}\) passive for the cases in which \(\bm{\mathcal{G}}\) represents GD (\(\gamma = \beta = 0\)), HB (\(\gamma = 0\)), and NAG (\(\gamma = \beta\)). The resulting lower bounds for any \(\alpha \in \mathbb{R}_{>0}\) and \(\beta \in [0, 1)\) are given in \Cref{tab:momentum_passivity_D}. Notably, for \(\beta = 0\), both \(D_{\mathrm{HB}}\) and \(D_{\mathrm{NAG}}\) reduce to \(D_{\mathrm{GD}}\), which is consistent with the fact that both HB and NAG reduce to GD when \(\beta = 0\).
    \begin{table}[t]
    \centering
    \vspace{10pt}
    \caption{Lower Bound on the Feedthrough Term \(D\) Needed to Render the Modified First-Order Optimization Controllers Passive.}
    \label{tab:momentum_passivity_D}
    \begin{tabularx}{\columnwidth}{%
        >{\centering\arraybackslash}p{0.18\columnwidth}
        >{\centering\arraybackslash}p{0.27\columnwidth}
        >{\centering\arraybackslash}p{0.45\columnwidth}
        }
        \hline\hline
        \addlinespace[2pt]
        \(\bar{\bm{\mathcal G}}_{\mathrm{GD}}\) 
        & 
        \(\bar{\bm{\mathcal G}}_{\mathrm{HB}}\) 
        & 
        \(\bar{\bm{\mathcal G}}_{\mathrm{NAG}}\)
        \\[2pt]
        \hline
        \addlinespace[3pt]
        \(
        D_{\mathrm{GD}} = \frac{\alpha}{2}
        \)
        &
        \(
        D_{\mathrm{HB}}=\frac{\alpha(1+\beta)}{2(1-\beta)^2}
        \)
        &
        \(
        D_{\mathrm{NAG}} = \alpha
        +
        \frac{\alpha(3\beta-1)}{2(1 - \beta)^2}
        \)
        \\\addlinespace[3pt]
        \hline\hline
    \end{tabularx}
    \vspace{-15pt}
    \end{table}
\subsection{Main Stability Result}
    This section presents a stability result based on the weak passivity theorem for the negative feedback interconnection in \Cref{fig:momentum_loop_transformation} and outlines the scope of the convergence guarantee it provides. In particular, \Cref{eqn:main_result} gives sufficient conditions for the shifted gradient to vanish asymptotically, but does not in general guarantee convergence of the iterates to the global minimizer. To strengthen this result, it is additionally assumed that \(f\) has a unique stationary point and that its gradient cannot become arbitrarily small far from this point. Under these assumptions, \Cref{eqn:main_result_convergence} establishes convergence of the iterates to the global minimizer.

    \begin{theorem}\label{eqn:main_result}
        Consider the first-order momentum-based method given by \cref{eqn:recursion_momentum_augmented} and its negative feedback representation in \Cref{fig:momentum_loop_transformation}, where
        \begin{equation}\label{eqn:lower_bound_on_D_for_stability}
                D =
            \begin{cases}
                \dfrac{\alpha \mleft(2 \gamma + 1\mright)}{2 \mleft(\beta + 1\mright)}, & \beta \leq \dfrac{\gamma}{\gamma + 1},\\[2ex]%
                \dfrac{\alpha }{\mleft(1 - \beta\mright)^2} - \dfrac{\alpha \mleft(2 \gamma + 1\mright)}{2 \mleft(1 - \beta\mright)}, & \dfrac{\gamma}{\gamma + 1} < \beta < 1.
            \end{cases}
        \end{equation}
        For \(f \in \mathcal{F}_{0, L}\), if \(\mbf{r}_1 = \mbf{0}\), \(\bar{\mbf{r}}_2 \in \ell_2\), and \(D < 1/L\), then \(\mbf{y}_{2}^{k} \to \mbf{0}\) as \(k \to \infty\). 
    \end{theorem}
    \begin{proof}
        From \Cref{lemma:passivity_of_G_bar} and \Cref{rmk:case_of_equality_for_beta_gamma}, if \(D\) satisfies \cref{eqn:lower_bound_on_D_for_stability}, then \(\bar{\bm{\mathcal{G}}}\) is passive. Moreover, from \Cref{lem:modified_nonlinearity_passivity}, if \(D < 1/L\), then \(\bar{\mbs{\Delta}}\) is OSP\@. Therefore, from \Cref{thrm:weak_passivity}, if \(\mbf{r}_1 = \mbf{0}\) and \(\bar{\mbf{r}}_2 \in \ell_2\), then \(\mbf{y}_2 \in \ell_2\), which implies \(\mbf{y}_2^k \to \mbf{0}\) as \mbox{\(k \to \infty\)}~\cite{feedback_systems}.\hspace*{\fill}~\QED\par\endtrivlist\unskip
    \end{proof}

    Since \Cref{eqn:main_result} relies on the weak passivity theorem, it only provides sufficient conditions for \(\mbf{y}_2^k \to \mbf{0}\) and makes no claims regarding the output of the modified controller \(\bar{\bm{\mathcal{G}}}\). To further highlight the subtleties of this result, notice that for \(\mbf{r}_1 = \bar{\mbf{r}}_2 = \mbf{0}\) in \Cref{fig:momentum_loop_transformation}, the result \(\mbf{y}_2 \in \ell_2\) implies \(\mbf{y}_2^k = \nabla f(\mbf{u}_2^k + \mbf{x}^\ast) = \nabla f(\mbf{y}_1^k + \mbf{x}^\ast) = \nabla f(\mbf{C}\mbs{\xi}^k + \mbf{x}^\ast) \to \mbf{0}\) as \(k \to \infty\). From \cref{eqn:augmented_state} and \cref{eqn:momentum_controller_minimal_realization}, it follows that
    \begin{align}\label{eqn:convergence_of_gradient}
        \nabla f (\mbf{x}^k + \gamma ( \mbf{x}^k - \mbf{x}^{k-1} ) )  \to \mbf{0}, \quad \text{as } k \to \infty.
    \end{align}
    When \(f\) has multiple minimizers, \cref{eqn:convergence_of_gradient} cannot guarantee the convergence of \(\mbf{x}^k\) to the specific global minimizer \(\mbf{x}^\ast\). Moreover, even if \(f\) has a unique minimizer \(\mbf{x}^\ast\) such that \(\nabla f(\mbf{x}^\ast) = \mbf{0}\), \cref{eqn:convergence_of_gradient} does not guarantee \(\mbf{x}^k \to \mbf{x}^\ast\) as \(k \to \infty\). The following example illustrates this point.
    \begin{example}
        Let \(f(x) = x^2/(x^2 + 1)\) for \(x \in \mathbb{R}\). Then, \(\nabla f(x) = 2x/(x^2 + 1)^2\) and \(\nabla f(x) = 0\) if and only if \(x = 0\). However, for \(z^k = k\), \(\nabla f(z^k) \to 0\) as \(k \to \infty\), but \(z^k \to \infty\). 
    \end{example}

    To address these issues, the following assumption is made on the function \(f\).
    \begin{assumption}\label{assumption1}
        There exists a unique stationary point \(\mbf{x}^\ast\), and far from this point, the gradient cannot become arbitrarily small, meaning \(\liminf_{\|\mbf{x}\|\to\infty} \mleft\| \nabla f(\mbf{x}) \mright\| > 0\).
    \end{assumption}
    The subset of functions in \(\mathcal{F}_{0, L}\) satisfying \Cref{assumption1} is denoted by \(\hat{\mathcal{F}}_{0, L}\). Importantly, for \(0 < m \leq L\), 
    \begin{equation}
        \mathcal{S}_{m, L} \subset \mathcal{F}_{m, L} \subset \hat{\mathcal{F}}_{0, L} \subset \mathcal{F}_{0, L}.
    \end{equation}
    The following corollary of \Cref{eqn:main_result} provides a stronger convergence result for functions in \(\hat{\mathcal{F}}_{0, L}\). The proof is provided in \hyperref[appendix:proof_of_convergence_for_S_hat]{Appendix B}.
    \begin{corollary}
        \label{eqn:main_result_convergence}
        Consider the first-order momentum-based method given by \cref{eqn:recursion_momentum_augmented} and its negative feedback representation in \Cref{fig:momentum_loop_transformation}, with \(D\) defined in \cref{eqn:lower_bound_on_D_for_stability}. For \(f \in \hat{\mathcal{F}}_{0, L}\), if \(\mbf{r}_1 = \bar{\mbf{r}}_2 = \mbf{0}\) and \(D < 1/L\), then \(\mbf{y}_2^k = \nabla f(\mbf{x}^k + \gamma ( \mbf{x}^k - \mbf{x}^{k-1} )) \to \mbf{0}\) and \(\mbf{x}^k \to \mbf{x}^\ast\) as \(k \to \infty\).
    \end{corollary}

    \section{Discussion}\label{sec:discussion}
In this section, the passivity-based analysis of momentum-based methods is compared to existing results in the literature. The results discussed below follow from applying \Cref{eqn:main_result_convergence} with the lower bound on \(D\) provided in \Cref{tab:momentum_passivity_D}. In this context, guaranteed convergence refers to the convergence of the iterates to the global minimizer \(\mbf{x}^\ast\). 

For the special case of \(\beta = 0\) and \(\gamma = 0\), the recursion in \cref{eqn:recursion_momentum_augmented} reduces to the GD method, and \Cref{eqn:main_result_convergence} guarantees convergence for any stepsize \(\alpha \in (0, 2/L)\). Notably, unlike the analyses in \cite[Theorem 3]{moalemi_forbes_gd_passive_acc} and \cite[Example 5.4]{simpson}, both of which recover the same stepsize interval for convergence of the GD method, \Cref{eqn:main_result_convergence} does not require a strictly positive lower sector bound. On the other hand, if \(f\) is only assumed to be \(L\)-smooth and bounded from below, then Polyak's analysis in~\cite[Section~1.4.2]{Polyak} guarantees \(\nabla f(\mbf{x}^k) \to \mbf{0}\) for \(\alpha \in (0, 2/L)\). However, it can be shown that with the addition of \Cref{assumption1}, Polyak's analysis also guarantees convergence of the iterates to the global minimizer for the same stepsize interval.

In the case of \(\beta > 0\) and \(\gamma = 0\), the recursion in \cref{eqn:recursion_momentum_augmented} reduces to the HB method, and \Cref{eqn:main_result_convergence} guarantees convergence for any \(\alpha > 0\) and \(\beta \in [0, 1)\) such that
\begin{align}
    \alpha < \frac{2(1-\beta)^2}{L(1+\beta)}. \label{eqn:HB_convergence_condition}
\end{align}
Existing related HB results in the sector-bounded framework assume \(f \in \mathcal{S}_{m, L}\) or \(f \in \mathcal{F}_{m, L}\) with \(m > 0\). In \cite{ugrinovskii}, for \(f \in \mathcal{F}_{m, L}\), the HB method is shown to be globally convergent for any \(\alpha > 0\) and \(\beta \in [0, 1)\) such that
\begin{align}\label{eqn:HB_convergence_condition_ugrinovskii}
    \alpha <
    \begin{cases}
        \dfrac{2(1 + \beta)}{L}, & \beta \leq \bar{\beta}, \\[1ex]
        \dfrac{2(1-\beta)^2}{(1+\beta)(L+m)-4\sqrt{\beta mL}}, & \beta > \bar{\beta},
    \end{cases}
\end{align}
where \(\bar{\beta} = \left( \sqrt{\tfrac{L}{m}} - \sqrt{\tfrac{L}{m} - 1}\right)^2\). Notice, if \(m \to 0\), then \(\bar{\beta} \to 0\) and the first case in \cref{eqn:HB_convergence_condition_ugrinovskii} reduces to \mbox{\(\alpha < 2/L\)}, with \(\beta = 0\), which is the same convergence condition as that of the GD method. On the other hand, the second case in \cref{eqn:HB_convergence_condition_ugrinovskii} reduces to \cref{eqn:HB_convergence_condition}. Thus, although the result in \cite{ugrinovskii} is established only for \(m>0\), its convergence condition formally reduces to \cref{eqn:HB_convergence_condition} in the limit as \(m \to 0\). In this sense, \Cref{eqn:main_result_convergence} provides an explicit convergence guarantee for the case \mbox{\(m=0\)}.

Finally, for the case of \(\beta > 0\) and \(\gamma = \beta\), the recursion in \cref{eqn:recursion_momentum_augmented} reduces to the NAG method, and \Cref{eqn:main_result_convergence} guarantees convergence for any \(\alpha > 0\) and \(\beta \in [0, 1)\) such that \(D_{\mathrm{NAG}} < 1/L\), which is equivalent to
\begin{align}
    \alpha < \frac{2(1-\beta)^2}{L\mleft( 2 \beta^2 - \beta + 1 \mright)}. \label{eqn:NAG_convergence_condition}
\end{align}
To the best of the authors' knowledge, there are no existing results that provide an explicit sufficient convergence region for the NAG method for functions in \(\hat{\mathcal{F}}_{0, L}\). 

    \section{Numerical Results} \label{sec:simulation}
\begin{figure}[t]
    \centering
    \vspace{5pt}
    \includegraphics[width=0.9\columnwidth]{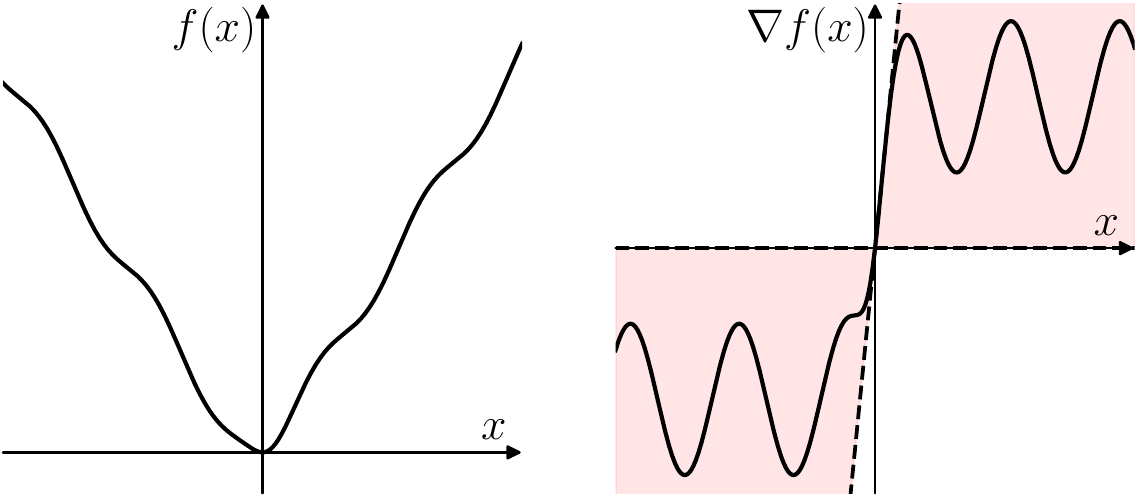}
    \caption{The function \(f \in \hat{\mathcal{F}}_{0, L}\) defined in \cref{example_function} with \(L\) given by \cref{eqn:L_for_example_function}.}
    \label{fig:sector_bound}
\end{figure}
\begin{figure}[t]
    \centering
    \includegraphics[width=0.95\columnwidth]{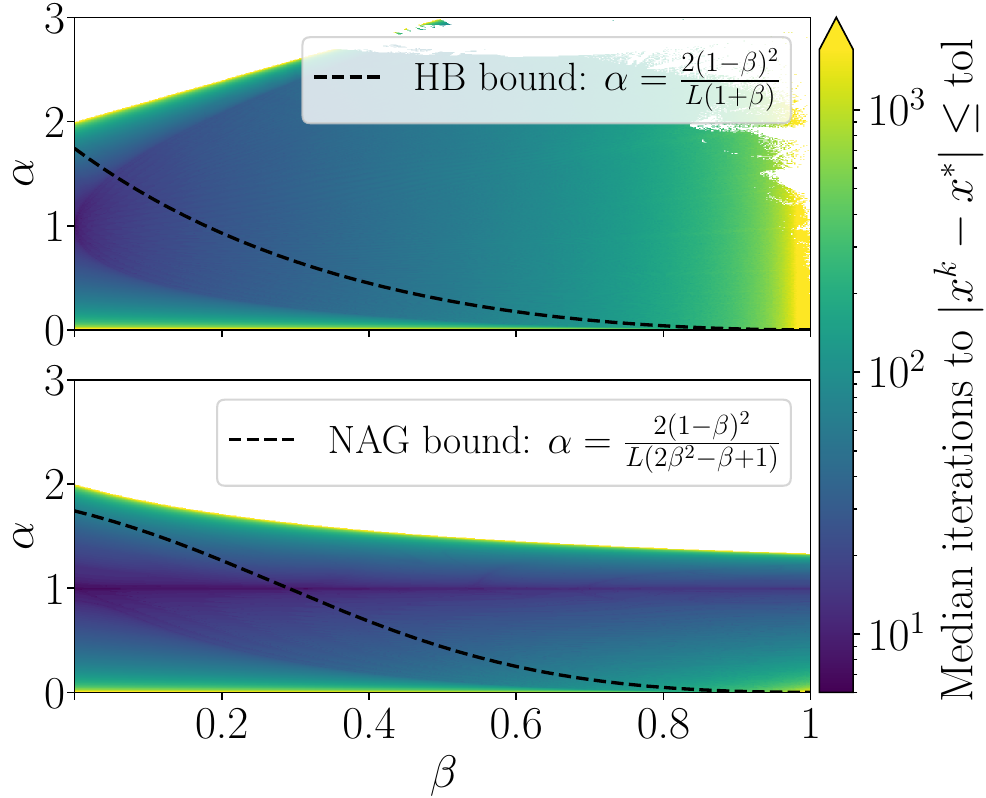}
    \vspace{-1pt}
    \caption{Heatmaps over the \((\beta,\alpha)\) plane for the HB (top) and NAG (bottom) methods applied to the test function in \Cref{fig:sector_bound}. Blank regions correspond to parameter values for which at least one initialization did not satisfy the stopping criterion within \(50{,}000\) iterations. The region below each dashed line is the guaranteed convergence region where \Cref{eqn:main_result_convergence} applies.}
    \label{fig:heatmaps}
\end{figure}
Consider the function \(f : \mathbb{R} \to \mathbb{R}\) defined by 
\begin{align}\label{example_function}
    f(x) &= \int_{0}^{x} \tanh(\tau) \mleft( 1 + \frac{\sin(\tau)}{2} \mright)  \mathrm{d}\tau.
\end{align}
Since \(f'(x) = \tanh(x) \mleft( 1 + \frac{\sin(x)}{2} \mright)\) and \(1 + \frac{\sin(x)}{2} > 0\) for all \(x \in \mathbb{R}\), it follows that \(f'(x) = 0\) if and only if \(x = x^\ast = 0\). Additionally, \(\liminf_{|x|\to\infty} |f'(x)| = \frac{1}{2} > 0\).
Therefore, \(f\) satisfies \Cref{assumption1}. As shown in \Cref{fig:sector_bound}, \(f\) is non-convex and \(f'(x)\) is in the sector \([0, L]\) with 
\begin{align}\label{eqn:L_for_example_function}
    L = \sup_{x \in \mathbb{R}} \frac{\tanh(x) \mleft( 1 + \frac{\sin(x)}{2} \mright)}{x} \approx 1.148.
\end{align}
Therefore, \(f \in \hat{\mathcal{F}}_{0, L}\), with \(L\) given by \cref{eqn:L_for_example_function}. As such, if the parameters \(\alpha\) and \(\beta\) of the HB and NAG methods satisfy \cref{eqn:HB_convergence_condition} and \cref{eqn:NAG_convergence_condition}, respectively, then \Cref{eqn:main_result_convergence} guarantees convergence of these methods to the unique minimizer \(x^\ast = 0\) from any initialization.

To assess how conservative these convergence conditions are, the HB and NAG methods were simulated over a grid of \(\alpha\) and \(\beta\) values. For each \((\alpha,\beta)\) pair, the methods were run from several logarithmically spaced initial conditions. A run was counted as successful if the iterates satisfied \(\|x^k-x^\ast\| \leq \mathrm{tol} = 10^{-9}\) within \(50{,}000\) iterations. Each grid point was then assigned the median iteration count across all initializations or left blank if any initialization failed to satisfy this criterion. As shown in \Cref{fig:heatmaps}, the convergence conditions provided by \Cref{eqn:main_result_convergence} are quite conservative for HB, but appear to be less conservative for NAG\@. Nevertheless, for both methods, the guaranteed convergence regions lie entirely within the empirically observed convergence regions.

    \section{Closing Remarks}\label{sec:closing_remarks}
The passivity-based analysis of first-order momentum-based methods presented in this paper is a two-step approach that separates the controller properties of the method from the function properties of the objective. In particular, \Cref{lemma:passivity_of_G_bar} provides a lower bound on the feedthrough term \(D\) that depends only on \(\alpha\), \(\beta\), and \(\gamma\), whereas \Cref{lem:modified_nonlinearity_passivity} provides the upper bound \(D<1/L\), which depends only on the sector bound of the gradient of \(f\). Consequently, for functions in \(\mathcal{F}_{0,L}\) and hyperparameters satisfying the conditions of \mbox{\Cref{eqn:main_result}}, the weak passivity theorem guarantees asymptotic vanishing of the shifted gradient. To strengthen this result, \Cref{assumption1} introduces the subclass \(\hat{\mathcal{F}}_{0,L}\), for which the objective has a unique stationary point and the gradient remains bounded away from zero at infinity. For \(f\in\hat{\mathcal{F}}_{0,L}\), \Cref{eqn:main_result_convergence} guarantees convergence of the iterates to the global minimizer. This broader function class distinguishes the present analysis from those in \cite{lessard_recht_iqc, hu_lessard, ugrinovskii, alex_petersen, hu_lessard_dissipativity_nag, triple_momentum, van2025fastest, moalemi_forbes_gd_passive_acc}, all of which assume either \(f \in \mathcal{S}_{m,L}\) or \(f \in \mathcal{F}_{m,L}\) with \(m>0\).

    \appendix
    \subsection{Proof of \texorpdfstring{\Cref{lemma:passivity_of_G_bar}}{Lemma 2}}\label{appendix:proof_of_passivity_of_G_bar}
Following \cite[Lemma 3]{anderson}, to show that \(\bar{\bm{\mathcal{G}}}\) can be made passive for a sufficiently large \(D\), it is sufficient to show that there exists a \(\mbf{P}=\mbf{P}^{\trans} \succ 0\) such that 
\begin{equation}\label{eqn:positive_real}
    \mbf{M} =
    \begin{bmatrix}
        \mbf{A}^{\trans} \mbf{P} \mbf{A} - \mbf{P} & \mbf{A}^{\trans} \mbf{P} \mbf{B} - \mbf{C}^{\trans} \\
        \mleft( \mbf{A}^{\trans} \mbf{P} \mbf{B} - \mbf{C}^{\trans} \mright)^{\trans}  & \mbf{B}^{\trans} \mbf{P} \mbf{B} - \mleft( \mbf{D} + \mbf{D}^{\trans} \mright)
    \end{bmatrix}
    \preceq 0,
\end{equation}
with
\begin{align*}
    \mbf{A} &= \mbf{A}_0 \otimes \eye =
    \begin{bmatrix}
        1 + \beta & -\beta \\
        1         & 0
    \end{bmatrix}\otimes \eye,\;\;\;
    \mbf{B} = \mbf{B}_0 \otimes \eye =
    \begin{bmatrix}
        \alpha \\
        0
    \end{bmatrix}\otimes \eye,\\
    \mbf{C} &= \mbf{C}_0 \otimes \eye =
    \begin{bmatrix}
        1 + \gamma & -\gamma
    \end{bmatrix}\otimes \eye, \;\;\;\;
    \mbf{D} = D\eye.
\end{align*}
To simplify the analysis, the scalar case of \(x \in \mathbb{R}\) is considered, since if \(\mbf{P}_0 \succ 0\) and \(\mbf{M}_0 \preceq 0\), then \(\mbf{P} = \mbf{P}_0 \otimes \eye \succ 0\) and \(\mbf{M} = \mbf{M}_0 \otimes \eye \preceq 0\) as well~\cite[Theorem~4.2.15]{Horn_Johnson_1991}. Let
\begin{equation}
    \mbf{P} = 
    \begin{bmatrix}
        p_{1} & p_{2} \\
        \ast  & p_{3}
    \end{bmatrix},
\end{equation}
where \(p_1, p_2, p_3 \in \mathbb{R}\) and \(\ast\) denotes the symmetric term \(p_2\). 

\noindent The matrix \(\mbf{P}\) is positive definite if and only if 
\begin{align}\label{eqn:P_being_pd}
    p_1 &> 0, & p_1 p_3 - p_2^2 &> 0.
\end{align}
By the principal-minor test~\cite{gantmacher1959theory}, \(\mbf{M} \preceq 0\) if and only if
\begin{subequations}
    {%
    \begin{align}
        2p_{2} + p_{3} + 2\beta p_{1} + 2\beta p_{2} + \beta^{2}p_{1} &\leq 0, \label{eqn:tm_passivity_row_1}
        \\%
        \beta^{2}p_{1}                                                &\leq p_{3}, \label{eqn:tm_passivity_row_2}
        \\%
        \alpha^{2}p_{1}                                               &\leq 2D, \label{eqn:tm_passivity_row_3}
        \\%
        p_{2} + p_{3} + \beta p_{1} + \beta p_{2}                     &= 0, \label{eqn:tm_passivity_row_4}
    \end{align}
    \begin{equation}
        \begin{split}
            &2D \left( 2p_{2} + p_{3} + 2\beta (p_{1} + p_{2}) + \beta^{2}p_{1} \right) + 1\\%
            &\quad\quad- \gamma \left( 2\alpha \left( p_1 + \beta p_1 + p_2 \right) - \gamma - 2 \right)  \\%
            &\quad\quad+ \alpha^{2}(p_{1}^{2} + p_{2}^{2} - p_{1}p_{3})\leq 2\alpha(p_{2} + p_{1} + \beta p_{1}),
        \end{split}
    \end{equation}
    \begin{equation}
        \gamma^2 - 2\alpha\beta\gamma p_1 + \alpha^{2}p_{1}p_{3}     \leq 2D \left( p_{3} - \beta^{2}p_{1} \right), \label{eqn:tm_passivity_row_6}
    \end{equation}
    \begin{equation}\label{eqn:tm_passivity_row_7}
        \begin{split}
            &\alpha^{2}(p_{1}^{2}p_{3} - p_{1}p_{2}^{2} - p_{1}p_{3}^{2} + p_{2}^{2}p_{3})\\%
            &\quad\quad- \beta^{2}p_{1} + p_{3} + 2 \gamma \left( p_2 + p_3 + \beta \left( p_1 + p_2 \right)\right)
            \\%
            &\quad\quad- 2\alpha\gamma \left( p_2^2 + p_1p_2 + p_1p_3 + p_2p_3 + \beta \left( p_1 + p_2 \right)^2  \right) \\%
            &\quad\quad- 2\alpha
            \left(
                \beta p_{1}p_{2} + \beta p_{1}p_{3} + p_{1}p_{3} + p_{2}p_{3}
            \right) \\%
            &\quad\quad+ 2D\beta^{2}\left(p_{1}^{2} + 2p_{1}p_{2} + p_{2}^{2}\right) \\%
            &\quad\quad+ 4D\beta \left(p_{1}p_{2} + p_{1}p_{3} + p_{2}^{2} + p_{2}p_{3}\right) \\%
            &\quad\quad+ 2D \left(p_{2}^{2} + 2p_{2}p_{3} + p_{3}^{2}\right)
            \leq 0.
        \end{split}
    \end{equation}
    }%
\end{subequations}
From \cref{eqn:P_being_pd}, \cref{eqn:tm_passivity_row_2}, and \cref{eqn:tm_passivity_row_4}, it can be shown that \(\beta^2 p_1 = p_3\) leads to a contradiction. As such, it follows that
\begin{equation}\label{eqn:tm_passivity_row_2_strict}
        \beta^{2}p_{1} < p_{3}.
\end{equation}
From \cref{eqn:tm_passivity_row_3}, consider the case where \(D = \alpha^2 p_1 / 2\). Assuming \(\beta < \frac{\gamma}{1 + \gamma} \), it can be shown that for
\begin{equation}\label{eqn:tm_passivity_conditions_for_beta_less_than_gamma_over_gamma_plus_1}
\begin{gathered}
\begin{aligned}
p_1 &= \frac{\gamma}{\alpha \beta}, &
p_2 &= \frac{\beta - \gamma}{\alpha \beta}, &
p_3 &= \frac{\gamma - \beta - \beta^2}{\alpha \beta}, &
D &= \frac{\alpha \gamma}{2 \beta},
\end{aligned}
\end{gathered}
\end{equation}
there exists a \(\mbf{P} \succ 0\) such that \(\mbf{M} \preceq 0\). Instead of assuming \(D = \alpha^2 p_1 / 2\), consider the case where \(\gamma^2 - 2\alpha\beta\gamma p_1 + \alpha^{2}p_{1}p_{3} = 2D \mleft( p_{3} - \beta^{2}p_{1} \mright)\) in \cref{eqn:tm_passivity_row_6}. It follows that 
\begin{equation}\label{eqn:tm_passivity_D_based_on_p1_p3}
    \frac{\gamma^2 - 2\alpha\beta\gamma p_1 + \alpha^{2}p_{1}p_{3}}{2(p_{3} - \beta^{2}p_{1})} = D.
\end{equation}
Moreover, from \cref{eqn:tm_passivity_row_4}, it follows that
\begin{equation}\label{eqn:tm_passivity_p_2_in_terms_of_p1_p3}
    p_2 = - \frac{p_3 + \beta p_1}{1 + \beta}.
\end{equation}
Substituting \cref{eqn:tm_passivity_D_based_on_p1_p3} and \cref{eqn:tm_passivity_p_2_in_terms_of_p1_p3} into \cref{eqn:tm_passivity_row_7} and completing the square yields
\begin{equation}\label{eqn:tm_passivity_row_7_simplified}
    \mleft( p_3 - \beta^2 p_1 \mright)\mleft( \beta - \alpha p_1 + \alpha p_3 + 1 \mright)^2 \leq 0.
\end{equation}
As per \cref{eqn:tm_passivity_row_2_strict}, \(p_3 - \beta^2 p_1 > 0\). Consequently, \cref{eqn:tm_passivity_row_7_simplified} implies that
\begin{equation}\label{eqn:tm_passivity_p1_in_terms_of_p3}
    p_1 = \frac{1 + \beta}{\alpha} + p_3.
\end{equation}
Assuming \(\beta < 1\), substituting \cref{eqn:tm_passivity_p1_in_terms_of_p3} into \cref{eqn:tm_passivity_row_2_strict}, it follows that
\begin{equation}\label{eqn:tm_passivity_p3_lower_bound}
    \frac{\beta^2}{\alpha \mleft( 1 - \beta \mright)} < p_3.
\end{equation}
Moreover, substituting \cref{eqn:tm_passivity_p1_in_terms_of_p3} into \cref{eqn:tm_passivity_D_based_on_p1_p3} yields
\begin{equation}\label{eqn:tm_passivity_D_based_on_p3}
    D = \frac{-\alpha^3p_3^2 + (2\beta\gamma - 1 - \beta)\alpha^2 p_3 + \mleft( 2\beta^2 + 2\beta - \gamma \mright)\alpha \gamma}{2 \mleft( (\beta^2 - 1)\alpha p_3 + \beta^3 + \beta^2 \mright)}.
\end{equation}
What remains is to find an expression for \(p_3\) such that \(D\) is minimized. To this end, the analysis is split into two cases based on the value of \(\beta\).

\noindent\textbf{Case 1:} $\beta < \frac{\gamma}{\gamma+1}$.

\noindent It can be shown that for \(p_3 = \frac{\gamma - \beta}{\alpha}\), the value of \(D\) in \cref{eqn:tm_passivity_D_based_on_p3} is minimized and given by
\begin{equation}
    D = \frac{\alpha(2\gamma + 1)}{2(\beta + 1)}.
\end{equation}
Notably, for \(\beta < \frac{\gamma}{\gamma+1}\), the value of \(D\) in \cref{eqn:tm_passivity_D_based_on_p3} is strictly smaller than that of \cref{eqn:tm_passivity_conditions_for_beta_less_than_gamma_over_gamma_plus_1}. As such, if \(\beta < \frac{\gamma}{1 + \gamma} \), for
\begin{align}\label{eqn:tm_passivity_conditions_case_1}
    p_1 &= \frac{1 + \gamma}{\alpha}, 
    & 
    p_2 &= -\frac{\gamma}{\alpha}, 
    & 
    p_3 &= \frac{\gamma - \beta}{\alpha}, 
    & 
    D &= \frac{\alpha(2\gamma + 1)}{2(\beta + 1)},
\end{align}
there exists a \(\mbf{P} \succ 0\) such that \(\mbf{M} \preceq 0\). 

\noindent\textbf{Case 2:} $\frac{\gamma}{\gamma+1} < \beta < 1$.

\noindent Similar to Case 1, it can be shown that for 
\begin{equation}
    p_3 = \frac{\beta - \gamma + \beta \gamma + \beta^2}{\alpha \mleft( 1 - \beta \mright) },
\end{equation}
the value of \(D\) in \cref{eqn:tm_passivity_D_based_on_p3} is minimized and given by
\begin{equation}
    D = \frac{\alpha }{\mleft(1 - \beta\mright)^2} - \frac{\alpha \mleft(2 \gamma + 1\mright)}{2 \mleft(1 - \beta\mright)}.
\end{equation}
As such, if \(\frac{\gamma}{\gamma+1} < \beta < 1\), for
\begin{equation}\label{eq:tm_passivity_conditions_case_2}
    \begin{gathered}
        p_1 = \frac{2}{\alpha(1-\beta)} - \frac{1+\gamma}{\alpha}, \quad
        p_2 = \frac{2+\gamma}{\alpha} - \frac{2}{\alpha(1-\beta)},\\
        p_3 = \frac{\beta - \gamma + \beta\gamma + \beta^2}{\alpha(1-\beta)}, \quad
        D   = \frac{\alpha}{(1-\beta)^2} - \frac{\alpha(2\gamma+1)}{2(1-\beta)},
    \end{gathered}
\end{equation}
there exists a \(\mbf{P} \succ 0\) such that \(\mbf{M} \preceq 0\). Combining \cref{eqn:tm_passivity_conditions_case_1} and \cref{eq:tm_passivity_conditions_case_2} yields the desired result.\hspace*{\fill}~\QED\par\endtrivlist\unskip

\subsection{Proof of \texorpdfstring{\Cref{eqn:main_result_convergence}}{Corollary 1}}\label{appendix:proof_of_convergence_for_S_hat}
As per \Cref{eqn:main_result}, for \(f \in \hat{\mathcal{F}}_{0, L} \subset \mathcal{F}_{0, L}\) and \(D\) defined in \cref{eqn:lower_bound_on_D_for_stability}, if \(\mbf{r}_1 = \bar{\mbf{r}}_2 = \mbf{0}\), then \(\mbf{y}_2^k \to \mbf{0}\) as \(k \to \infty\). Moreover, if \(\mbf{r}_1 = \bar{\mbf{r}}_2 = \mbf{0}\) in \Cref{fig:momentum_loop_transformation}, then \(\mbf{y}_2^k = \nabla f(\mbf{u}_2^k + \mbf{x}^\ast) = \nabla f(\mbf{y}_1^k + \mbf{x}^\ast) = \nabla f(\mbf{C}\mbs{\xi}^k + \mbf{x}^\ast)\). From \cref{eqn:augmented_state} and \cref{eqn:momentum_controller_minimal_realization}, it follows that
\begin{align}
    \nabla f (\mbf{z}^k )  \to \mbf{0}, \quad \text{as } k \to \infty,
\end{align}
where \(\mbf{z}^k = \mbf{x}^k + \gamma ( \mbf{x}^k - \mbf{x}^{k-1} )\). Since \(f \in \hat{\mathcal{F}}_{0, L} \), from \Cref{assumption1}, it follows that \(\liminf_{\|\mbf{z}\|\to\infty} \mleft\| \nabla f(\mbf{z}) \mright\| > 0\). As such, there exist \(r > 0\) and \(c > 0\) such that \(\mleft\| \nabla f(\mbf{z}) \mright\| \geq c\) for all \(\mleft\| \mbf{z} \mright\| \geq r\). For a fixed \(\varepsilon > 0\), define the compact set
\begin{align}
    \mathcal{R}_{\varepsilon} = \mleft\{ \mbf{z} \in \mathbb{R}^n \mid \mleft\| \mbf{z} \mright\| \leq r, \mleft\| \mbf{z} - \mbf{x}^\ast \mright\| \geq \varepsilon \mright\}. 
\end{align}
Since \(\mbf{x}^\ast\) is the unique stationary point of \(f\), \(\nabla f(\mbf{z})\) is continuous, and \(\mathcal{R}_{\varepsilon}\) is compact, it follows that 
\begin{align}
    m_{\varepsilon} = \min_{\mbf{z} \in \mathcal{R}_{\varepsilon}} \mleft\| \nabla f(\mbf{z}) \mright\| > 0.
\end{align}
Consequently, every point \(\mbf{z}\) with \(\mleft\| \mbf{z} - \mbf{x}^\ast \mright\| \geq \varepsilon\) satisfies \(\mleft\| \nabla f(\mbf{z}) \mright\| \geq c_{\varepsilon} > 0\), where \(c_{\varepsilon} = \min \mleft\{ m_{\varepsilon}, c \mright\}\). Equivalently, every \(\mbf{z} \in \mathbb{R}^{n}\) such that \(\mleft\| \nabla f(\mbf{z}) \mright\| < c_{\varepsilon}\) satisfies \mbox{\(\mleft\| \mbf{z} - \mbf{x}^\ast \mright\| < \varepsilon\)}. Since \(\nabla f(\mbf{z}^k) \to \mbf{0}\), it follows that \(\mbf{z}^k \to \mbf{x}^\ast\) as \(k \to \infty\).

Next, to show that \(\mbf{x}^k \to \mbf{x}^\ast\) as \(k \to \infty\), let \(\mbf{d}^k = \mbf{x}^k - \mbf{x}^{k-1}\). From \cref{eqn:momentum_controller_minimal_realization}, it follows that 
\begin{align}
    \mbf{d}^{k+1} = \beta \mbf{d}^k - \alpha \nabla f(\mbf{z}^k).
\end{align}
Since \(\nabla f(\mbf{z}^k) \to \mbf{0}\) and \(\beta < 1\), it follows that \(\mbf{d}^k \to \mbf{0}\) as \(k \to \infty\). Finally, from the definition of \(\mbf{z}^k\), it follows that \(\mbf{x}^k = \mbf{z}^k - \gamma \mbf{d}^k\). Since \(\mbf{z}^k \to \mbf{x}^\ast\) and \(\mbf{d}^k \to \mbf{0}\), it follows that \(\mbf{x}^k \to \mbf{x}^\ast\) as \(k \to \infty\).\hspace*{\fill}~\QED\par\endtrivlist\unskip

    \printbibliography
\end{document}